\documentclass[11pt]{article}
\pdfoutput=1
\usepackage[round]{natbib}
\usepackage{amsmath}
\usepackage{amsthm}
\usepackage{amsfonts}
\usepackage{bm}
\usepackage{graphicx}
\usepackage{hyperref}  
\usepackage[margin = 1in]{geometry}

\def\va{{\bm{a}}}
\def\vw{{\bm{w}}}
\def\vm{{\bm{m}}}
\def\vc{{\bm{c}}}
\def\vx{{\bm{x}}}
\def\vy{{\bm{y}}}
\def\vb{{\bm{b}}}

\def\ve{{\bm{e}}}
\def\vzero{{\bm{0}}}
\def\vz{{\bm{z}}}
\newcommand{\bb}[1]{\mathbb{#1}}
\newcommand{\tG}{\ensuremath{\tilde G}}
\newtheorem{theorem}{Theorem}
\DeclareMathOperator*{\minimize}{minimize}
\newcommand{\opt}{\ensuremath{\mathrm{opt}}}
\newcommand{\fra}{\ensuremath{\mathrm{frac}}}
\newcommand{\bin}{\ensuremath{\mathrm{bin}}}
\newtheorem{fact}[theorem]{Fact}
\newtheorem{lemma}[theorem]{Lemma}

\DeclareMathOperator*{\argmin}{arg\,min}

\newcommand\myeq{\stackrel{\mathrm{def}}{=}}
\DeclareMathOperator*{\argmax}{arg\,max}

\renewcommand{\top}{{\ensuremath{\mathsf T}}}
\newcommand{\truth}{{\ensuremath{\textup{truth}}}}

\newcommand{\tkmeans}{\ensuremath{T_{k\textup{-means}}}}

\title{Approximation Algorithms for Orthogonal Non-negative Matrix Factorization}
\author{Moses Charikar 
\thanks{Computer Science Department, Stanford University.
Email: \texttt{moses@cs.stanford.edu}.
Supported by a Simons Investigator Award, a Google Faculty Research Award and an Amazon Research Award.
} 
\and 
Lunjia Hu 
\thanks{Computer Science Department, Stanford University. 
Email: \texttt{lunjia@stanford.edu}.
Supported by NSF Award IIS-1908774 and a VMware fellowship.}}
\date{}

\allowdisplaybreaks

\begin{document}
\maketitle
\begin{abstract}%
In the non-negative matrix factorization (NMF) problem, the input is an $m\times n$ matrix $M$ with non-negative entries and the goal is to factorize it as $M\approx AW$. The $m\times k$ matrix $A$ and the $k\times n$ matrix $W$ are both constrained to have non-negative entries. This is in contrast to singular value decomposition, where the matrices $A$ and $W$ can have negative entries but must satisfy the orthogonality constraint: the columns of $A$ are orthogonal and the rows of $W$ are also orthogonal. The orthogonal non-negative matrix factorization (ONMF) problem imposes both the non-negativity and the orthogonality constraints, and previous work showed that it leads to better performances than NMF on many clustering tasks. We give the first constant-factor approximation algorithm for ONMF when one or both of $A$ and $W$ are subject to the orthogonality constraint. We also show an interesting connection to the correlation clustering problem on bipartite graphs. Our experiments on synthetic and real-world data show that our algorithm achieves similar or smaller errors compared to previous ONMF algorithms while ensuring perfect orthogonality (many previous algorithms do not satisfy the hard orthogonality constraint).
\end{abstract}
\section{Introduction}
Low-rank approximation of matrices is a fundamental technique in data analysis. Given a large data matrix $M$ of size $m\times n$, 
the goal is
to approximate it by a low-rank matrix $AW$ where $A$ has size $m\times k$ and $W$ has size $k\times n$. Here $k$ is called the \emph{inner dimension} of the factorization $M \approx AW$, controlling the rank of $AW$. 
Such low-rank matrix decomposition enables a succinct and often more interpretable representation of the original data matrix $M$.

One of the standard approaches of low-rank approximation is singular value decomposition (SVD) \citep{wold1987principal,alter2000singular, papadimitriou2000latent}. SVD computes a solution minimizing both the Frobenius norm $\|M-AW\|_F$ and the spectral norm $\sigma_{\textnormal{max}}(M-AW)$ \citep{eckart1936approximation,mirsky1960symmetric}. In addition, SVD always gives a solution with the orthogonality property: the columns of $A$ are orthogonal and the rows of $W$ are also orthogonal. Orthogonality makes the factors more separable, and thus causes the low-rank representation to have a cleaner structure. 

However, in certain cases the data matrix $M$ is inherently non-negative, with entries corresponding to frequencies or probability mass, and in these cases SVD has a serious limitation: the factors $A$ and $W$ computed by SVD often contain negative entries, making the factorization less interpretable. Non-negative matrix factorization (NMF), which constrains $A$ and $W$ to have non-negative entries, is better suited to these cases, and is applied in many domains including computer vision \citep{lee1999learning, li2001learning}, text mining \citep{xu2003document, pauca2004text} and bioinformatics \citep{brunet2004metagenes, kim2007sparse,devarajan2008nonnegative}.

One drawback of NMF relative to SVD is that it gives less separable factors: the angle between any two columns of $A$ or any two rows of $W$ is at most $\pi/2$ simply because the inner product of a pair of vectors with non-negative coordinates is always non-negative. 
To reap the benefits of non-negativity and orthogonality simultaneously, orthogonal NMF (ONMF)
adds orthogonality constraints to NMF on one or both of the factors $A$ and $W$: the columns of $A$ and/or the rows of $W$ are required to be orthogonal. Indeed, ONMF leads to better empirical performances in many clustering tasks \citep{ding2006orthogonal,choi2008algorithms,yoo2010nonnegative}.
While previous works showed ONMF algorithms that converge to local minima \citep{ding2006orthogonal} 
and an efficient polynomial-time approximation scheme (EPTAS) assuming the inner-dimension is a constant \citep{asteris2015orthogonal},
a theoretical understanding of the worst-case guarantee one can achieve for ONMF with arbitrary inner-dimension is lacking.
In this work, we show the first constant-factor approximation algorithm for ONMF 
with respect to the squared Frobenius error $\|M - AW\|_F^2$
when the orthogonality constraint is imposed on one or both of the factors.

\paragraph{Our Results} 
We use approximation algorithms for weighted $k$-means as subroutines, such as the $(9+\varepsilon)$-approximation local search algorithm by \citet{kanungo2002local}. Assuming an $r$-approximation algorithm for weighted $k$-means, we show algorithms for ONMF with approximation ratio $2r$ in the \emph{single-factor orthogonality} setting where only one of the factors $A$ or $W$ is required to be orthogonal (Theorem \ref{thm:single}), and approximation ratio $\left(2r + \frac{8r + 8}{\sin^2(\pi/12)}\right)$ in the \emph{double-factor orthogonality} setting where both $A$ and $W$ are required to be orthogonal (Theorem \ref{thm:main}). Here, $A$ (resp.\ $W$) being orthogonal means that its columns (resp.\ rows) are orthogonal but not necessarily of unit length. The approximation ratios are provable upper bounds for the ratio between the error of the output $(A, W)$ of the algorithm and the minimum error over all feasible solutions $(A, W)$, with error measured using the squared Frobenius norm $\|M - AW\|_F^2$.
We also demonstrate the superior practical performance of our algorithms by experiments in both the single-factor and the double-factor orthogonality setting on synthetic and real-world datasets (see Section~\ref{sec:experiments} and Appendix~\ref{sec:additional-exp}).

\paragraph{Sparse Structure of Solution} 
When we impose the orthogonality constraint on both the columns of $A$ and the rows of $W$,
the non-negativity and the orthogonality constraints together cause the solution to ONMF to have a very sparse structure. 
Let $\va_i$ denote the $i$-th column of $A$ and $\vw_i^\top$ denote the $i$-th row of $W$. Since $\va_i$ and $\va_j$ are constrained to have non-negative entries but zero inner product, they have disjoint supports, and this also holds for $\vw_i$ and $\vw_j$. As a result, $AW = \sum_{i=1}^k\va_i\vw_i^\top$ naturally consists of $k$ disjoint blocks, as shown in Figure \ref{fig}. 

If the input $M$ factorizes as $M = AW$ exactly, we can easily recover $A$ and $W$ based on the block-wise structure of $M$. Therefore, we focus on the agnostic setting where $M=AW$ does not hold exactly, and design approximation algorithms that find solutions comparable to the best possible factorization.

\begin{figure}[h]
\label{fig}
\centering
\includegraphics[width = 0.7\textwidth]{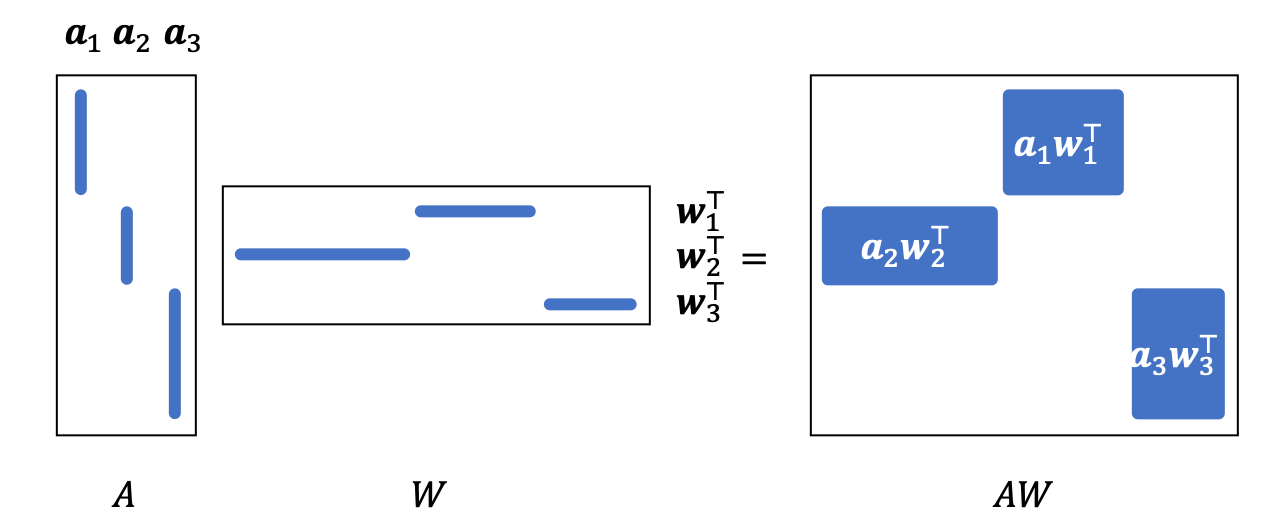}
\caption{The $k$ columns of $A$ have disjoint supports. The $k$ rows of $W$ also have disjoint supports. The product $AW$ has entries equal to zero outside the $k$ blocks.}
\end{figure}

\paragraph{Connection to Bipartite Correlation Clustering} %

The block-wise structure of $AW$ (Figure \ref{fig}) relates ONMF to the correlation clustering problem \citep{bansal2004correlation} on complete bipartite graphs.

To see the relationship with correlation clustering, let us consider a data matrix $M$ with binary entries and assume $k\geq\min\{m,n\}$. Since we can find at most $\min\{m,n\}$ non-zero $\va_i\vw_i^\top$ satisfying the orthogonal constraint, all $k\geq\min\{m,n\}$ give equivalent problems, where any inner-dimension is considered feasible. $M$ can be treated as a complete bipartite graph with vertices $\{u_1,\ldots,u_m\}\cup\{v_1,\ldots,v_n\}$ and edges $(u_i,v_j)$ labeled ``$+$'' if $M_{ij} = 1$ or ``$-$'' if $M_{ij} = 0$. This edge-labeled complete bipartite graph is exactly an instance of the correlation clustering problem. If the factors $A$ and $W$ also have binary entries and both satisfy the orthogonality constraint, the blocks of $AW = \sum_{i=1}^k\va_i\vw_i^\top$ (see Figure \ref{fig}) are all-ones matrices corresponding to vertex-disjoint complete bipartite sub-graphs. This is exactly the form of a solution to the correlation clustering problem, and the objective $\|M-AW\|_F^2$ is exactly the number of disagreements in the correlation clustering problem. Although our algorithm (specifically, the algorithm in Theorem \ref{thm:k>=n}) doesn't impose the binary constraint on $A$ and $W$, we can apply the following lemma to each block of $AW$ to round the solution to binary with only a constant loss in the objective (see Appendix \ref{sec:rounding} for proof): 
\begin{lemma}
\label{lm:round}
Let $M\in\{0,1\}^{m\times n}$ be a binary matrix. Let $\va\in\bb R^m_{\geq 0}$ and $\vw\in\bb R^n_{\geq 0}$ be two non-negative vectors. Then, there exist binary vectors $\hat\va\in\{0,1\}^m$ and $\hat\vw\in\{0,1\}^n$ such that
\begin{equation*}
    \|M-\hat\va\hat\vw^\top\|_F^2\leq
    8\|M-\va\vw^\top\|_F^2.
\end{equation*}
Moreover, $\hat\va$ and $\hat\vw$ can be computed in poly-time.
\end{lemma}
Thus, we can obtain an approximation algorithm for minimizing disagreements in complete bipartite graphs via our approximation algorithm for ONMF in Theorem \ref{thm:k>=n}.
Moreover, without the binary constraint on $M,A,W$, 
ONMF with orthogonality constraint on both $A$ and $W$ can be treated as a soft version of bipartite correlation clustering.

\paragraph{Open Questions} We used the Frobenius norm as a natural measure of goodness of fit, but it would be interesting to see if one can achieve constant-factor approximation with respect to other measures, such as the spectral norm, since the two norms can be different by a factor that grows with $\min\{m,n\}$. It would also be interesting to consider replacing the orthogonality constraint on $A$ and $W$ by a lower bound $\theta < \pi/2$ on the angles between different columns of $A$ and different rows of $W$.

\paragraph{Related Work} Non-negative matrix factorization was first proposed by \citet{paatero1994positive}, and was shown to be NP-hard by \citet{vavasis2010complexity}. Algorithmic frameworks for efficiently finding local optima include the multiplicative updating framework \citep{lee2001algorithms} and the alternating non-negative least squares framework \citep{lin2007projected, kim2011fast}. Under the usually mild \emph{separability} assumption, \citet{arora2016computing} showed an efficient algorithm that computes the global optimum.

\cite{ding2006orthogonal} first studied NMF with the orthogonality constraint, and showed its effectiveness in document clustering. After that, algorithms for ONMF using various techniques have been developed for a broad range of applications \citep{chen2009collaborative,ma2010orthogonal,kuang2012symmetric,pompili2013onp,li2014discriminative,kim2015mutation,qin2016community,alaudah2017weakly,huang2019collaborative}.
The less restrictive single-factor orthogonality setting attracted the most attention, and most algorithms for solving it belong to the multiplicative updating framework: iteratively updating $A$ and/or $W$ by taking the element-wise product with other computed non-negative matrices
\citep{yang2007multiplicative,choi2008algorithms,yoo2008orthogonal,yoo2010nonnegative,yang2010linear,pan2018orthogonal,he2020low}.
Other techniques include HALS (hierarchical alternating least squares) \citep{li2014two,kimura2016column} and using a penalty function \citep{del2009penalty} for the orthogonality constraint.

While improving the separability of the factors compared to NMF, these algorithm {\em do not guarantee convergence to a solution that has perfect orthogonality} (which is also demonstrated in our experiments).
There are only a few previous algorithms that have this guarantee, including the EM-ONMF algorithm \citep{pompili2014two}, the ONMFS algorithm \citep{asteris2015orthogonal} and the NRCG-ONMF algorithm \cite{zhang2016efficient}. ONMFS is the only previous algorithm we know that has a provable approximation guarantee, but it has a running time exponential in the squared inner dimension.
\citep{pompili2014two} give a reduction of ONMF to spherical $k$-means with a somewhat non-standard objective function:
the goal is to minimize the sum of $1$ minus the square of cosine similarity, while the commonly studied objective function for spherical $k$-means sums up $1$ minus the cosine similarity.
Our results for ONMF imply a constant factor approximation for this variant of spherical $k$-means with the squared cosine similarity in the objective.
Many variants of ONMF have also been studied in the literature, including the semi-ONMF \citep{li2018semi} and the sparse ONMF \citep{chen2018soft,li2020two}.

We would also like to point out that the connection between ONMF and $k$-means shown in \citep[Theorems 1 and 2]{ding2006orthogonal} does not give a reduction in either direction. Their proof shows that the optimization problem associated with $k$-means is essentially ONMF, but \emph{with additional constraints}: the matrix $G$ in the ONMF formulation (8) in \citep{ding2006orthogonal} is replaced by matrix $\tG$ in the $k$-means formulation (11) in \citep{ding2006orthogonal}. However $\tG$ is a “normalized cluster indicator matrix” that is more constrained than the generic matrix $G$ with orthonormal columns because the entries in every column of $\tG$ are either zero or take the same non-zero value. This additional constraint makes their argument insufficient to either directly derive an algorithm for ONMF with the same approximation guarantee given one for $k$-means, or the other way around. Also, later works such as \cite{yoo2010nonnegative} and \cite{asteris2015orthogonal} used techniques different from $k$-means to improve the empirical performance of ONMF.

The correlation clustering problem was proposed by \citet{bansal2004correlation} on complete graphs, who showed a constant factor approximation algorithm for the disagreement minimization version and a polynomial-time approximation scheme (PTAS) for the agreement maximization version. \citet{ailon2008aggregating} showed a simple combinatorial algorithm achieving an approximation ratio of 3 in the disagreement minimization version, and \citet{chawla2015near} improved the approximation ratio to the currently best 2.06. \citet{chawla2015near} also showed a $3$-approximation algorithm on complete $k$-partite graphs.

\section{Weighted \texorpdfstring{$k$}{k}-Means}
\label{sec:preliminaries}
The $k$-means problem is a fundamental clustering problem, and we will apply algorithms for its weighted version as subroutines to solve our orthogonal NMF problem. Given points $\vm_1,\ldots,\vm_n\in\bb R^m$ and their weights $\ell_1,\ldots,\ell_n\in\bb R_{\geq 0}$, the weighted $k$-means problem seeks $k$ centroids $\vc_1,\ldots,\vc_k$ and an assignment mapping $\phi:\{1,\ldots,n\}\rightarrow\{1,\ldots,k\}$ that solve the following optimization problem:
\begin{equation*}
\minimize_{\vc_1,\ldots,\vc_k;\phi} ~ \sum_{i=1}^n\ell_i\|\vm_i - \vc_{\phi(i)}\|_2^2.
\end{equation*}
Even the unweighted ($\forall i,\ell_i = 1$) version of this problem is APX hard, but many constant factor approximation algorithms were obtained. \citet{kanungo2002local} showed a local-search algorithm achieving an approximation ratio $9+\varepsilon$,\footnote{The algorithm of \citet{kanungo2002local} was originally designed for the unweighted setting, but it works naturally in the weighted setting if we use the algorithm by \citet{feldman2007ptas} when computing the  \emph{$(k,\varepsilon)$-approximate centroid set} on which local search is performed.} which was improved by \cite{ahmadian2017better} in the unweighted setting to an approximation ratio $6.357$.
\section{Single-factor Orthogonality}
\label{sec:single}
In the single-factor orthogonality setting, we impose the orthogonality constraint only on one of the factors $A$ or $W$. For concreteness, let us assume that the rows of $W$ are required to be orthogonal. 
Since the rows of $W$ are also non-negative, they must have disjoint supports, 
or equivalently, each column of $W$ has at most one non-zero entry. This particular structure relates our problem closely to the weighted $k$-means problem, and it's not hard to apply the approximation algorithms for weighted $k$-means to our single-factor orthogonality setting. %
Specifically, assuming there is a poly-time $r$-approximation algorithm for weighted $k$-means, we show a poly-time algorithm for the single-factor orthogonality setting with approximation factor $2r$ (Theorem \ref{thm:single}). %

To see why $k$-means plays an important role in our problem, recall that the non-negativity and orthogonality constraints on $W$ simplify each column $\vw_i$ of $W$ to the form $\theta_i\ve_{\phi(i)}$, where $\theta_i$ is a non-negative real number, $\phi$ maps $\{1,\ldots,n\}$ to $\{1,\ldots,k\}$, and $\ve_{\phi(i)}\in\bb R^k$ is the unit vector with its $\phi(i)$-th coordinate being one. This means that the $i$-th column of $AW$ is exactly $\theta_i$ times the $\phi(i)$-th column of $A$. If we think of the $k$ columns of $A$ as $k$ centroids, and $\phi$ as the assignment mapping that maps every column of $M$ to its closest centroid, (unweighted) $k$-means is exactly our problem with the additional constraint that $\theta_i = 1$ for all $i$.

With the freedom of choosing $\theta_i$, it's more convenient to solve our problem by \emph{weighted} $k$-means. 
Assume
without loss of generality that 
every column of $A$ in the optimal solution is 
the zero vector or has unit length
as we can always scale them back using $\theta_i$.
We normalize the columns of $M$ and weight each column proportional to its initial squared $L_2$ norm.
After that, always setting $\theta_i = 1$ only increases the approximation ratio by a factor of 2 as we show in the following lemma proved in Appendix \ref{sec:proof-unit} (think of $\vx$ as a column of the optimal $A$ and $\vy$ as a column of $M$):
\begin{fact}
\label{fact:unit}
Let $\vx\in\bb R^m_{\geq 0}$ be a unit vector or the zero vector. For any non-negative vector $\vy\in\bb R^m_{\geq 0}$ and any $\theta\geq 0$, we have $\|\vy-\theta\vx\|_2^2\geq\frac 12 \|\vy\|_2^2\cdot \|\bar \vy - \vx\|_2^2$, where $\bar \vy = \left\{\begin{array}{ll}\frac{\vy}{\|\vy\|_2}, & \vy\neq\vzero\\ \vzero,&\vy = \vzero\end{array}\right.$.
\end{fact}
Based on this intuition, we obtain the following algorithm.
Let $\vm_1,\vm_2,\ldots,\vm_n\in\bb R^m_{\geq 0}$ be the columns of $M$, and let $\bar\vm_i$ be the normalized version of $\vm_i$: 
\begin{equation*}
\bar\vm_i := \left\{\begin{array}{ll}\frac{\vm_i}{\|\vm_i\|_2},& \textup{if }\vm_i\neq \vzero\\ \vzero,& \textup{if }\vm_i = \vzero\end{array}\right..
\end{equation*}
Let $\ell_i:=\|\vm_i\|_2^2$ be the \emph{weight} of point $\bar \vm_i\in\bb R^m_{\geq 0}$. We first compute an $r$-approximate solution to the following weighted $k$-means problem:
\begin{equation}
\label{eq:k-means}
    \minimize\limits_{\vc_1,\ldots,\vc_k;\phi} ~  \sum_{i=1}^n\ell_i\|\bar \vm_i-\vc_{\phi(i)}\|_2^2.
\end{equation}
We can assume WLOG that all of the centroids $\vc_1,\ldots,\vc_k$ have non-negative coordinates since increasing the negative coordinates to zero never increases the weighted $k$-means objective. Then we simply output $A = [\vc_1,\ldots,\vc_k]$ and $W = [\theta_1\ve_{\phi(1)},\ldots, \theta_n\ve_{\phi(n)}]$, where
\begin{equation*}
    \theta_i \! = \! \left\{\begin{array}{ll}\!\!\!\frac{\langle \vm_i,\vc_{\phi(i)}\rangle}{\|\vc_{\phi(i)}\|_2^2}, & \textup{\!\!\!\!\! if }\vc_{\phi(i)}\neq \vzero\\ \!\!\!0,& \textup{\!\!\!\!\! if }\vc_{\phi(i)} = \vzero\end{array}\right.\!\!\!\in\!\argmin_\theta \|\vm_i - \theta\vc_{\phi(i)}\|_2^2.
\end{equation*}

We show the approximation guarantee in the following theorem proved in Appendix \ref{sec:proof-single}.
\begin{theorem}
\label{thm:single}
The algorithm above computes
a $2r$-approximate solution $A$ and $W$ in the single-factor orthogonality setting in time $O(\tkmeans + mn)$, where $\tkmeans$ is the time needed by the weighted $k$-means subroutine.
\end{theorem}

\section{Double-factor Orthogonality}
\label{sec:double}
Now we consider the double-factor orthogonality setting, where we require $A$ to have orthogonal columns \emph{and} $W$ to have orthogonal rows, 
and show a poly-time constant factor approximation algorithm in this setting.

We first state some basic facts that will be used in the discussion of our algorithms.
\paragraph{Useful Inequalities}
The following \emph{doubled triangle inequality} for the squared $L_2$ distance between vectors $\vx$ and $\vy$ is useful when we analyze the approximation ratio of our algorithm:
\begin{fact}
\label{fact:doubled_triangle}
$\|\vx - \vy\|_2^2\leq 2\|\vx\|_2^2 + 2\|\vy\|_2^2$.
\end{fact}
When both $\vx$ and $\vy$ have non-negative coordinates, we have the following stronger fact:
\begin{fact}
\label{fact:non-negative_triangle}
If both $\vx$ and $\vy$ have non-negative coordinates, then $\|\vx-\vy\|_2^2\leq \|\vx\|_2^2 + \|\vy\|_2^2$.
\end{fact}
\paragraph{Center of Mass}
Given $n$ points $\vx_1,\ldots,\vx_n\in\bb R^m$ and their weights $\ell_1,\ldots,\ell_n\in\bb R_{\geq 0}$, the point $\vy\in\bb R^m$ minimizing the weighted sum of the squared $L_2$ distances $\sum_{i=1}^n\ell_i\|\vx_i-\vy\|_2^2$ is the center of mass: $\vy = \left(\sum_{i=1}^n\ell_i\vx_i\right) / \left(\sum_{i=1}^n\ell_i\right)$. Moreover, the weighted sum can be decomposed using the following identity (see, for example, Lemma 2.1 in \citep{kanungo2002local}):
\begin{fact}
\label{fact:center}
Assume $\ell_1,\ldots,\ell_n\geq 0$ and $\vy=\left(\sum_{i=1}^n\ell_i\vx_i\right) / \left(\sum_{i=1}^n\ell_i\right)$. Then for any vector $\vb$, we have 
\begin{equation*}
\sum_{i=1}^n\ell_i\|\vx_i-\vb\|_2^2 = \sum_{i=1}^n\ell_i\|\vx_i-\vy\|_2^2 + \sum_{i=1}^n\ell_i\|\vy-\vb\|_2^2.
\end{equation*}
\end{fact}

\subsection{Intuition}
We describe the intuition that leads us to the algorithm.
As our first step, we solve the weighted $k$-means problem as we did in the single-factor orthogonality setting,
but we need to additionally ensure that the columns of $A$ are orthogonal.
By the \emph{doubled triangle inequality} (Fact \ref{fact:doubled_triangle}) and the property of the center of mass (Fact \ref{fact:center}), we can move the $n$ points to their centroids without affecting the approximation ratio too much. Now there are only $k$ distinct points, and it's more convenient to treat these points as vectors, so that we can talk about the angles between them. Our goal is to find $k$ \emph{orthogonal} centroids that approximate these $k$ vectors. The key challenge is to find the assignment mapping: which vectors are mapped to the same centroid, and once we know the assignment mapping, we can find the best centroids by optimizing each coordinate separately (see (\ref{eq:final_optimization})). 
Intuitively, the assignment mapping should respect the angles between the vectors: if a pair of vectors form a ``small'' angle, they should be mapped to the same centroid, and if they form a ``large'' angle close to $\pi/2$, they should be mapped to different centroids. However, two vectors both forming ``small'' angles with the third may themselves form a relatively ``large'' angle. In order to solve the lack of transitivity, we need to eliminate angles that are neither very ``small'' nor very ``large''. We make the observation that if the angle between two vectors is in the range $[\pi/6, \pi/3]$, they can't be simultaneously close to a set of orthonormal vectors, and thus they can't have low cost in the optimal solution, so we can safely ``ignore'' them by decreasing their weights by the same amount. This \emph{weight reduction} procedure eventually makes the angle between any two vectors lie in the range $[0,\pi/6)\cup(\pi/3, \pi/2]$. If two vectors both have angles less than $\pi/6$ with the third, they themselves cannot form an angle larger than $\pi/3$, so now we have the desired transitivity. Our Lemma \ref{lm:prune} shows that the assignment mapping computed this way is comparable to the optimal one.

\subsection{Algorithm}
\label{sec:algorithm}
Our algorithm consists of three major steps. The first step is to apply the weighted $k$-means algorithm as we did in the single-factor orthogonality setting, and two additional steps are needed to make sure the solution has both factors being orthogonal.
\subsubsection*{Step 1: Weighted \texorpdfstring{$k$}{k}-Means}
Let $\vm_1,\vm_2,\ldots,\vm_n\in\bb R^m_{\geq 0}$ be the columns of $M$ and define $\bar\vm_i$ and $\ell_i$ the same way as in Section \ref{sec:single}. Compute an $r$-approximate solution $\vc_1,\ldots,\vc_k, \phi$ to the weighted $k$-means problem (\ref{eq:k-means}). Define the weight $q_j$ of a centroid $\vc_j$ to be the total weight of the points assigned to it: $q_j:=\sum_{i\in\phi^{-1}(j)}\ell_i$. By Fact \ref{fact:center}, we can always assume WLOG that whenever $q_j>0$, it holds that $\vc_j = \left(\sum_{i\in\phi^{-1}(j)}\ell_i\bar\vm_i\right)/q_j$. Under this assumption, whenever $q_j>0$, we have $\|\vc_j\|_2\leq 1$. We also have the following easy fact:
\begin{fact}
\label{fact:nonzero}
If $q_j>0$, then $\vc_j\neq \vzero$.
\end{fact}
\begin{proof}
Assume for the sake of contradiction that $\vc_j = \vzero$. According to our assumption, we have $\vzero = \vc_j = \left(\sum_{i\in\phi^{-1}(j)}\ell_i\bar\vm_i\right)/q_j$, so for all $i\in\phi^{-1}(j)$, $\ell_i\bar\vm_i = \vzero$. If $\bar\vm_i\neq\vzero$, we know $\ell_i = 0$; otherwise, we know $\vm_i = \vzero$ and thus, again, $\ell_i = \|\vm_i\|_2^2 = 0$. Now we have our desired contradiction: $q_j = \sum_{i\in\phi^{-1}(j)}\ell_i = 0$.
\end{proof}
\subsubsection*{Step 2: Weight Reduction}
Recall that the weight $q_j$ of a centroid $\vc_j$ was defined to be the total weight of the points assigned to it.
The second step of the algorithm is to \emph{reduce} the weights $q_1,\ldots,q_k$ to $q_1',\ldots,q_k'$. To start, all $q_j'$ are initialized to be $q_j$. Our algorithm iterates over all pairs $(j_1,j_2)$ satisfying $1\leq j_1<j_2\leq k$. If $q_{j_1}'>0,q_{j_2}'>0$ and $\angle(\vc_{j_1},\vc_{j_2})\in [\pi/6,\pi/3]$, our algorithm decreases both $q_{j_1}',q_{j_2}'$ by the minimum of the two (thus sending at least one of them to 0). Recall Fact \ref{fact:nonzero} that $\vc_{j_1}$ and $\vc_{j_2}$ are both non-zero, so the angle between them is well-defined. 

\subsubsection*{Step 3: Finalize the Solution}
Now we are most interested in centroids $\vc_j$ with positive weights ($q_j'>0$) after the weight reduction step. For any two centroids $\vc_{j_1},\vc_{j_2}$ with positive weights, we know $\angle(\vc_{j_1},\vc_{j_2})\in[0,\pi/6)\cup(\pi/3,\pi/2]$. Since the angles between vectors satisfy the triangle inequality, we can group these centroids so that 
$\angle(\vc_{j_1},\vc_{j_2})\in[0,\pi/6)$ if $j_1,j_2$ belong to the same group, and
$\angle(\vc_{j_1},\vc_{j_2})\in(\pi/3,\pi/2]$ if $j_1,j_2$ belong to different groups.
Suppose $\vc_j$ belongs to group $\sigma(j)\in\{1,\ldots,k\}$.

We claim that we can find an \emph{optimal} solution to the following optimization problem in poly-time:
\begin{align}
    \minimize_{\va_1,\ldots,\va_k} ~ & \sum_{j:q'_j>0}q_j'\|\vc_j-\va_{\sigma(j)}\|_2^2,\nonumber \\
    \textup{s.t.\ } & \va_1,\ldots,\va_k\in\bb R^m_{\geq 0},\nonumber \\ 
    & \forall 1\leq s <t \leq k, \va_s^\top \va_t=0.\label{eq:final_optimization}
\end{align}
To solve the above optimization problem, we decompose it \emph{coordinate-wise}. Specifically, the constraints on $\va_1,\ldots,\va_k$ can be translated to that for every $h\in\{1,\ldots,m\}$, the $h$-th coordinates $a_{1,h},\ldots,a_{k,h}$ are all non-negative and contain at most one positive value. The objective can also be decomposed coordinate-wise:
\begin{align*}
    \sum_{j:q'_j>0}q_j'\|\vc_j-\va_{\sigma(j)}\|_2^2 & = \sum_{h=1}^m\sum_{j:q'_j>0}q_j'(c_{j,h}-a_{\sigma(j),h})^2 \\
    & \myeq \sum_{h=1}^m O_h.
\end{align*}
If we define $q^*_s$ as the total weight of the $s$-th group: $q^*_s = \sum_{j\in\sigma^{-1}(s)}q'_j$, and when $q^*_s> 0$ define $\mu_{s,h}$ as the weighted average of the $h$-th coordinate of the centroids in the $s$-th group: $\mu_{s,h} = (q^*_s)^{-1}\sum_{j\in\sigma^{-1}(s)}q_j'c_{j,h}$, we can further decompose the objective above using Fact \ref{fact:center} as
\begin{equation*}
O_h = \sum_{j:q'_j>0}q_j'(c_{j,h}-\mu_{\sigma(j),h})^2 + \sum_{s:q^*_s > 0} q^*_s (\mu_{s,h} - a_{s,h})^2.
\end{equation*}
The first term does not depend on $\va_1,\ldots,\va_k$, and the second term is minimized when $a_{s,h} = \mu_{s,h}$ for $s = \argmax_sq_s^*\mu_{s,h}^2$ and $a_{s,h} = 0$ for other $s$.
We have thus computed the optimal solution to \eqref{eq:final_optimization}.
Since $\|\vc_j\|_2\leq 1$ whenever $q_j'>0$, it is straightforward to check that $\|\va_s\|_2\leq 1$ for $s = 1,\ldots,k$.

We output $A = [\va_1,\ldots,\va_k]$ and $W = [\theta_1\ve_{\sigma(\phi(1))}, \ldots, \theta_n\ve_{\sigma(\phi(n))}]$ as the final solution, where
\begin{align*}
    \theta_i & = \left\{\begin{array}{ll}\frac{\langle \vm_i,\va_{\sigma(\phi(i))}\rangle}{\|\va_{\sigma(\phi(i))}\|_2^2}, & \textup{if }\va_{\sigma(\phi(i))}\neq \vzero\\ 0,& \textup{if }\va_{\sigma(\phi(i))} = \vzero\end{array}\right.\\
    & \in\argmin_\theta \|\vm_i - \theta\va_{\sigma(\phi(i))}\|_2^2.
\end{align*}
Note that $\sigma(j)$ was defined only for $j$ with $q_j'>0$, but here we extend its definition to all $j\in\{1,\ldots,k\}$ by setting the remaining values arbitrarily.

\subsection{Analysis}
We show the following two theorems on the approximation guarantee of our algorithm in the double-factor orthogonality setting. Theorem \ref{thm:main} applies to general inner dimensions $k$, while Theorem \ref{thm:k>=n} gives improved approximation factors when $k$ is large, which is the case when we apply our ONMF algorithm to correlation clustering. Recall that we used an $r$-approximation algorithm for weighted $k$-means as a subroutine, and we assume that its running time is $\tkmeans$. 
\begin{theorem}
\label{thm:main}
The algorithm in Section \ref{sec:algorithm} computes a $\left(2r + \frac{8r + 8}{\sin^2(\pi/12)}\right)$-approximate solution $A$ and $W$ in the double-factor orthogonality setting in time $O(\tkmeans + mn + mk^2)$.
\end{theorem}
\begin{theorem}
\label{thm:k>=n}
When $k\geq \min\{m,n\}$, there exists an algorithm that gives a $\frac{1}{\sin^2(\pi/12)}(\leq 15)$-approximate solution in the double-factor orthogonality setting in time $O(mn^2)$.
\end{theorem}
We prove Theorem \ref{thm:main} based on the following lemma, which we prove in Appendix \ref{sec:proof-prune}. We defer the proof of Theorem \ref{thm:k>=n} to Appendix \ref{sec:proof-k>=n}.
\begin{lemma}
\label{lm:prune} 
Let $\vz_1,\ldots,\vz_{k_1}\in\bb R^m_{\geq 0}$ be non-negative unit vectors that are orthogonal to each other. For any $\sigma':\{1,\ldots,k\}\rightarrow\{1,\ldots,k_1\}$, we have
\begin{equation*}
\sum_{j=1}^kq_j\|\vc_j-\va_{\sigma(j)}\|_2^2\leq \frac{2}{\sin^2(\pi/12)}\sum_{j=1}^kq_j\|\vc_j-\vz_{\sigma'(j)}\|_2^2.
\end{equation*}
\end{lemma}
\begin{proof}[Proof of Theorem \ref{thm:main}]
We obtain the running time of the algorithm by summing over the three steps.
Step 1 requires $O(mn)$ time to create the input to the weighted $k$-means subroutine, and the subroutine takes $\tkmeans$ time.
Step 2 takes $O(mk^2)$ time because we use $O(m)$ time to compute the angle between each of the $O(k^2)$ pairs of centroids.
In step 3, it takes $O(mk)$ time to solve the optimization problem \eqref{eq:final_optimization}, and it takes time $O(mn)$ to compute the $\theta_i$'s.

The feasibility of $(A,W)$ is clear from the algorithm. We focus on proving the approximation guarantee. We start by showing an upper bound for the objective $\|M-AW\|_F^2$ achieved by our algorithm. For $i=1,\ldots,n$, the $i$-th column of $AW$ is $\theta_i\va_{\sigma(\phi(i))}$, where $\theta_i\in\argmin_\theta\|\vm_i - \theta\va_{\sigma(\phi(i))}\|_2^2$. Therefore,
\begin{align*}
    & \|M-AW\|_F^2\\
    = {} & \sum_{i=1}^n\big\|\vm_i - \theta_i\va_{\sigma(\phi(i))} \big\|_2^2\\
    \leq {} & \sum_{i=1}^n\big\|\vm_i - \|\vm_i\|_2\va_{\sigma(\phi(i))} \big\|_2^2\\
    = {} & \sum_{i=1}^n\|\vm_i\|_2^2\cdot \|\bar\vm_i - \va_{\sigma(\phi(i))}\|_2^2.
\end{align*}
By Fact \ref{fact:center} and $\vc_j = \left(\sum_{i\in\phi^{-1}(j)}\ell_i\bar\vm_i\right)/q_j$, we have
\begin{align}
    &\|M-AW\|_F^2\nonumber\\
    \leq {}  &\sum_{i=1}^n\|\vm_i\|_2^2\cdot \|\bar\vm_i - \va_{\sigma(\phi(i))}\|_2^2\nonumber \\
    = {} & \sum_{i=1}^n\ell_i\cdot \|\bar\vm_i - \va_{\sigma(\phi(i))}\|_2^2\nonumber \\
    = {} & \sum_{i=1}^n\ell_i\cdot \|\bar\vm_i - \vc_{\phi(i)}\|_2^2 + \sum_{i=1}^n\ell_i\cdot \|\vc_{\phi(i)} - \va_{\sigma(\phi(i))}\|_2^2\nonumber \\
    = {} & \sum_{i=1}^n \ell_i \cdot \|\bar\vm_i - \vc_{\phi(i)}\|_2^2 + \sum_{j=1}^k q_j\|\vc_j-\va_{\sigma(j)}\|_2^2.\label{eq:alg}
\end{align}
(\ref{eq:alg}) gives an upper bound for $\|M-AW\|_F^2$. We proceed by giving a lower bound for the objective $\|M-A^\opt W^\opt\|_F^2$ achieved by the optimal solution $(A^\opt, W^\opt)$. We first remove the columns of $A^\opt$ filled with the zero vector and also remove the corresponding rows in $W^\opt$. This doesn't change the product $A^\opt W^\opt$ and doesn't violate the orthogonality requirement either, but the sizes of $A^\opt$ and $W^\opt$ may now change to $m\times k_1$ and $k_1\times n$. We can now assume WLOG that every column $\va^\opt_s$ of $A^\opt$ is a unit vector. Note that each column of $W$ contains at most one non-zero entry, so we have
\begin{align}
&\|M-A^\opt W^\opt\|_F^2\nonumber \\
\geq {} & \sum_{i=1}^n \min_{\genfrac{}{}{0pt}{2}{1\leq s\leq k_1}{\theta\geq 0}}\|\vm_i - \theta\va^\opt_s\|_2^2\nonumber\\
\geq {} & \frac 12\sum_{i=1}^n\|\vm_i\|_2^2\cdot \min_{1\leq s\leq k_1}\|\bar\vm_i - \va^\opt_s\|_2^2\nonumber\\
= {} & \frac 12\sum_{i=1}^n\ell_i\cdot \min_{1\leq s\leq k_1}\|\bar\vm_i - \va^\opt_s\|_2^2,\label{eq:aopt}
\end{align}
where the second inequality follows from Fact \ref{fact:unit}. By the $r$-approximate optimality of $\vc_1,\ldots,\vc_k$, we have
\begin{align}
    & \|M-A^\opt W^\opt\|_F^2\nonumber\\
    \geq {} &\frac 12\sum_{i=1}^n\ell_i\cdot \min_{1\leq s\leq k_1}\|\bar\vm_i - \va^\opt_s\|_2^2 \nonumber\\
    \geq {} & \frac 1{2r} \sum_{i=1}^n\ell_i\cdot \|\bar\vm_i - \vc_{\phi(i)}\|_2^2.\label{eq:c}
\end{align}
Combining (\ref{eq:aopt}) with (\ref{eq:c}), we have
\begin{align}
    & (4r+4)\|M-A^\opt W^\opt\|_F^2\nonumber\\
    \geq &  \sum_{i=1}^n\ell_i(2\|\bar\vm_i - \vc_{\phi(i)}\|_2^2 + 2\min_{1\leq s\leq k_1}\|\bar\vm_i - \va^\opt_s\|_2^2)\nonumber \\
    \geq & \sum_{i=1}^n\ell_i\min_{1\leq s\leq k_1}\|\vc_{\phi(i)} - \va^\opt_s\|_2^2\label{eq:3}\\
    = & \sum_{i=1}^n\ell_i\|\vc_{\phi(i)} - \va^\opt_{\sigma'(\phi(i))}\|_2^2\nonumber\\
    = & \sum_{j=1}^kq_j\|\vc_j - \va^\opt_{\sigma'(j)}\|_2^2,\nonumber
\end{align}
where (\ref{eq:3}) is by Fact \ref{fact:doubled_triangle} and $\sigma'(j)$ is defined to be $\argmin_{1\leq s\leq k_1}\|\vc_j - \va^\opt_s\|_2$. Applying Lemma \ref{lm:prune}, we get
\begin{align}
\label{eq:b}
    & (4r+4)\|M-A^\opt W^\opt\|_F^2\nonumber\\
    \geq {} & \sum_{j=1}^kq_j\|\vc_j - \va^\opt_{\sigma'(j)}\|_2^2\\
    \geq {} & \frac{\sin^2(\pi/12)}{2}\sum_{j=1}^kq_j\|\vc_j - \va_{\sigma(j)}\|_2^2.
\end{align}
Combining (\ref{eq:alg}) with (\ref{eq:c}) and (\ref{eq:b}), we have
\begin{align*}
    & \|M-AW\|_F^2\\
    \leq & \sum_{i=1}^n \ell_i \cdot \|\bar\vm_i - \vc_{\phi(i)}\|_2^2 + \sum_{j=1}^k q_j\|\vc_j-\va_{\sigma(j)}\|_2^2\\
    \leq & \left(2r + \frac{8r + 8}{\sin^2(\pi/ 12)}\right)\|M - A^\opt W^\opt\|_F^2.
\end{align*}
\end{proof}
\section{Experiments}
\label{sec:experiments}
We report on the results of experiments comparing the performance of our algorithm with eight previous algorithms in the literature. For these experiments, we use $k$-means++ as the subroutine for solving $k$-means.
For the single factor orthogonality setting,
our experiments show that our algorithm ensures perfect orthogonality and give similar approximation error as six previous algorithms in the literature that do not guarantee orthogonality.
For this single factor setting, we also directly compare to two previous algorithms that do ensure orthogonality and find that that the performance of our algorithm is superior. One of the previous algorithms has runtime that scales very poorly with inner dimension (and worse error for small inner dimension); the other suffers from poor local minima, leading to large error even with zero noise.
For the double factor orthogonality setting, only two previous algorithms are able to handle this case. None of them ensure perfect orthogonality, while our algorithm does. Further, it has lower error than these previous algorithms.
Our algorithm runs significantly faster than all these other algorithms in both settings.
Thus we achieve the best of both worlds -- stronger approximation guarantees as well as superior practical performance for ONMF.

Specifically, we compare our algorithm (ONMF-apx) with previous algorithms in the more well-studied single-factor orthogonality setting on synthetic data, and defer the experiments on real-world data and in the double-factor orthogonality setting to Appendix \ref{sec:additional-exp}. The previous algorithms we compare with include NMF \citep{lee2001algorithms}, PNMF \citep{yuan2005projective}, ONFS-Ding \citep{ding2006orthogonal}, NHL \citep{yang2007multiplicative}, ONMF-A \citep{choi2008algorithms}, HALS \citep{li2014two}, EM-ONMF \citep{pompili2014two}, and ONMFS \citep{asteris2015orthogonal}.

\paragraph{Experimental Setup}
We generate the input matrix $M\in\bb R^{m\times n}$ by adding noise to the product $M_\truth$ of random non-negative matrices $A_\truth\in\bb R^{m\times k}$ and $W_\truth\in\bb R^{k\times n}$. We make sure that $W_\truth$ has orthogonal rows\footnote{Due to non-negativity, making the rows of $W_\truth$ orthogonal is equivalent to making every column of $W_\truth$ contain at most one non-zero entry. Independently for every column, we pick the location of the non-zero entry uniformly at random.}, and every non-zero entry of $A_\truth$ and $W_\truth$ is independently drawn from the exponential distribution with mean 1. We call $M_\truth = A_\truth W_\truth$ \emph{the planted solution}, and we add iid noise to every entry of $M_\truth$ to obtain $M$. The noise also follows an exponential distribution, and we use the phrase ``noise level'' to denote the mean of that distribution.

\paragraph{Evaluation}
We measure the quality of the matrices $A$ and $W$ output by the algorithms in terms of the approximation error and the orthogonality of $W$.
We measure the approximation error using the Frobenius norm: we compute both the \emph{recovery error} $\|M_\truth - AW\|_F$, which measures how well the output recovers the underlying structure of the input, and the \emph{reconstruction error} $\|M - AW\|_F$, which measures the approximation error to the input matrix that contains iid noise. 
We define the reconstruction error of the planted solution $M_\truth$ as $\|M - M_\truth\|_F$, whose value concentrates well around $\sqrt{2mn}$ times the noise level as shown in the following easy fact:
\begin{fact}
\label{fact:planted-reconstruction}
The mean (resp.\ standard deviation) of $\|M - M_\truth\|_F^2$ is $2mn$ (resp.\ $\sqrt{20mn}$) times the noise level squared.
\end{fact}

We measure the non-orthogonality of $W$ by the Frobenius norm of $WW^\top - I$ after removing the zero rows of $W$ and normalizing the other rows.

\paragraph{Experiment 1}
In the first experiment, we choose $m = 100, n = 5000, k = 10$, and compare our algorithm with previous ones. We run each algorithm independently for 7 times and record the median results in Figure~\ref{fig:exp1}. We found that ONMFS could not finish in a reasonable amount of time, so  we investigate it separately on smaller matrices in experiment 2. We also found that there is a high variance in the approximation error of EM-ONMF because it often converges to a bad local optimum, giving the fluctuating black lines in Figure~\ref{fig:exp1}.
\begin{figure*}[h]
\centering
\includegraphics[width = 0.48\textwidth]{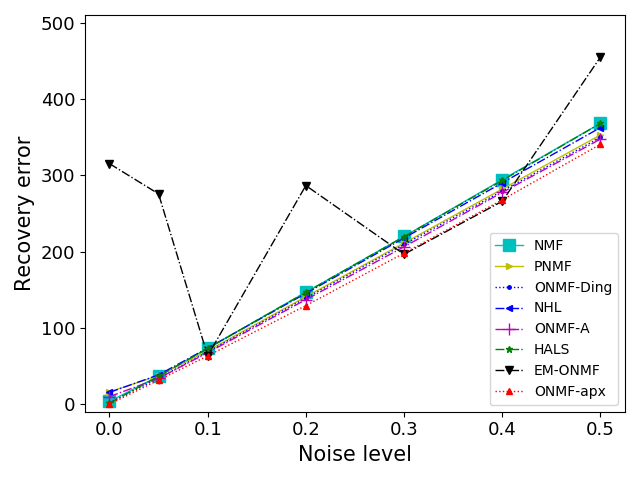}
\includegraphics[width = 0.48\textwidth]{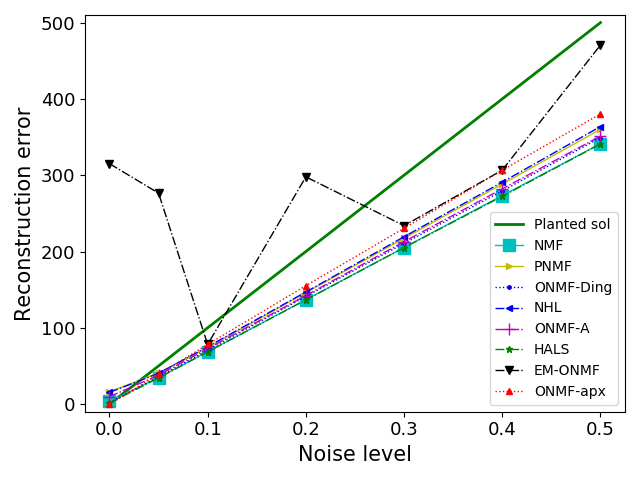}
\includegraphics[width = 0.48\textwidth]{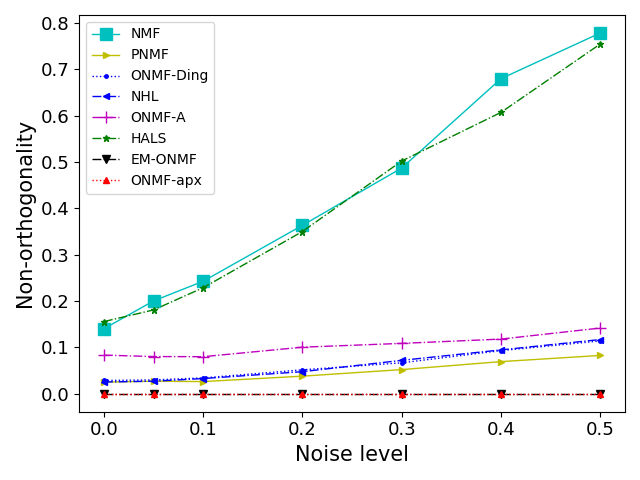}
\includegraphics[width = 0.48\textwidth]{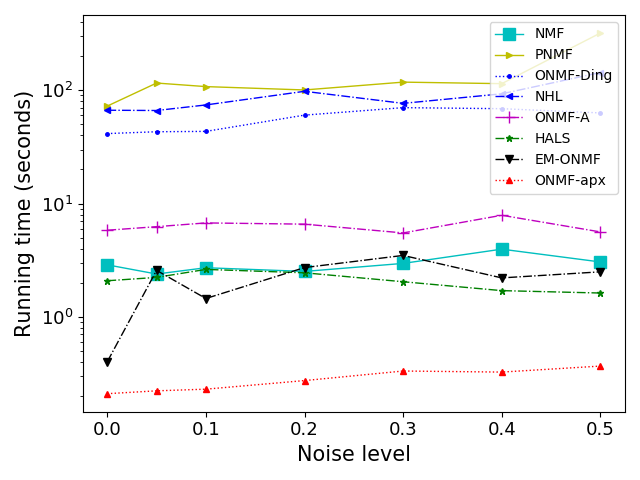}
\caption{Results of experiment 1. From left to right, the plots in the first row show the recovery error and the reconstruction error, and the plots in the second row show the non-orthogonality and the running time. The performance of our algorithm is shown in the red line under the label ONMF-apx.}
\label{fig:exp1}
\end{figure*}

As shown in Figure~\ref{fig:exp1}, our algorithm ensures perfect orthogonality and gives similar approximation error as previous ones which do not guarantee orthogonality. Except for EM-ONMF, none of the other previous algorithms in this experiment output a perfectly orthogonal $W$. Our recovery error is slightly better than previous algorithms, but our reconstruction error is slightly worse. This is because the orthogonality constraint effectively regularizes our solution, making it fit the noise in the input worse but reveal the structure of the input better.
It is worth noting that our algorithm achieves lower reconstruction errors than the planted solution $M_\truth$, and so do most other algorithms in the experiment (the reconstruction error of $M_\truth$ concentrates well around $1000$ times the noise level (thick green line in Figure~\ref{fig:exp1}) by Fact~\ref{fact:planted-reconstruction}).

We would also like to point out that our algorithm runs significantly faster than all the other algorithms considered in this experiment. The bottom right plot of Figure~\ref{fig:exp1} shows the running time on a machine with 1.4 GHz Quad-Core Intel Core i5 processor and 8 GB 2133 MHz LPDDR3 memory (note that the $y$-axis is on logarithmic scale). Our algorithm is based on the $k$-means++ subroutine, which is very efficient. The previous algorithms are based on iterative update and may take a long time to reach a local optimum. 

\paragraph{Experiment 2}
We compare our algorithm with ONMFS \citep{asteris2015orthogonal}, an algorithm that guarantees perfect orthogonality, but runs in time exponential in the squared inner dimension. ONMFS was based on two levels of exhaustive search, which is inefficient when the inner dimension is large. We thus reduce the sizes of the matrices and set $m = 10, n = 50, k = 2$ in this experiment. Our result shows that our algorithm gives smaller error than ONMFS (Figure \ref{fig:exp2}).
\begin{figure*}[h]
\centering
\includegraphics[width = 0.48\textwidth]{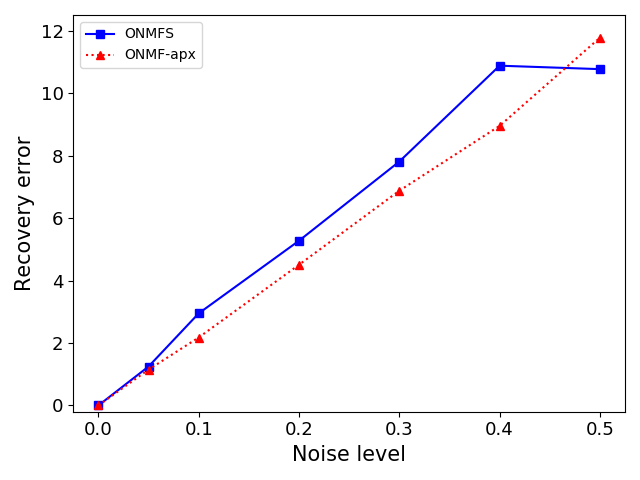}
\includegraphics[width = 0.48\textwidth]{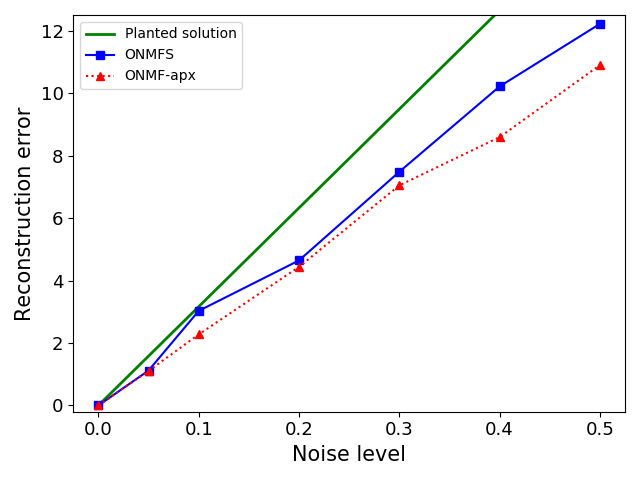}
\caption{Results of experiment 2. From left to right, the plots show the recovery error and the reconstruction error. The non-orthogonality (not shown in figure) is identically zero for both algorithms.}
\label{fig:exp2}
\end{figure*}

\section*{Acknowledgments}
We thank Suyash Gupta and anonymous reviewers for helpful comments on earlier versions of this paper.
\bibliographystyle{plainnat}
\bibliography{ref}
\newpage
\appendix
\section{Proof of Lemma~\ref{lm:round}}
\label{sec:rounding}
Before proving Lemma \ref{lm:round}, we first show how it gives a constant-factor approximation for bipartite correlation clustering. Given a complete bipartite graph with vertex bipartition $U\cup V$ and edges labeled $+$ or $-$, we can construct a binary matrix $M$ whose rows correspond to the vertices in $U$ and columns correspond to the vertices in $V$. An entry of $M$ is $1$ if and only if the corresponding edge is labeled $+$. 
The optimal solution to the correlation clustering problem also gives a binary matrix, where each cluster in the solution gives an all-ones block. Because of the block-wise structure, the matrix can be written in the form $A^\opt W^\opt$, where $A^\opt,W^\opt$ give a feasible solution to the orthogonal non-negative factorization problem for $M$ with the inner-dimension being the number of clusters in the optimal solution, and the squared Frobenius error $E^\opt:=\|M - A^\opt W^\opt\|_F^2$ equals to the optimal number of disagreements for the correlation clustering problem.

By Theorem \ref{thm:k>=n}, we can compute an orthogonal non-negative factorization $A^\fra W^\fra$ with inner dimension $k = \min\{|U|,|V|\}$ such that $E^\fra := \|M - A^\fra W^\fra\|_F^2 \leq 15 E^\opt$.
Note that $E^\fra$ can be decomposed block-wise:
\begin{equation*}
E^\fra := E^\fra_1 + \cdots + E^\fra_k + E^\fra_*,
\end{equation*}
where $E^\fra_i$ is the sum of squared errors in block $i$, and $E^\fra_*$ is the sum of squared errors outside of the $k$ blocks. Applying Lemma \ref{lm:round} to every block, we can compute an orthogonal non-negative factorization $A^\bin W^\bin$ such that every entry of $A^\bin$ and $W^\bin$ are binary, and the sum of squared errors in each block satisfies $E^\bin_i \leq 8E^\fra_i$. We also have $E^\bin_* = E^\fra_*$ because both $A^\bin W^\bin$ and $A^\fra W^\fra$ have zeros outside the $k$ blocks. Summing up, we have
\begin{align*}
E^\bin := {} & \|M - A^\bin W^\bin\|_F^2\\
= {} & E^\bin_1 + \cdots + E^\bin_k + E^\bin_*\\
\leq {} & 8\, E^\fra_1 + \cdots + 8\, E^\fra_k + E^\fra_*\\
\leq {} & 8 E^\fra\\
\leq {} & 8\cdot 15 E^\opt.
\end{align*}
Thus, if we translate every block of $A^\bin W^\bin$ to a cluster of vertices, we get a $8\cdot 15 = 120$ approximate solution to the correlation clustering problem.

We now return to proving Lemma \ref{lm:round}.
\begin{proof}
Write $\vw$ as $(w_1,w_2,\ldots,w_n)^\top$ and $\hat \vw$ as $(\hat w_1,\hat w_2,\ldots,\hat w_n)^\top$. Let $\vm_i$ denote the $i$-th column of $M$. We will construct $\hat\va,\hat\vw$ so that $\|\vm_i-\hat w_i\hat \va\|_2^2\leq 8\|\vm_i-w_i\va\|_2^2$ holds true for all $i$. If some $w_i=0$, we can always set $\hat w_i=0$. Therefore, without loss of generality, we can assume every $w_i$ is non-zero. Let $i^*\in\argmin_{1\leq i\leq n}\|\frac{\vm_i}{w_i}-\va\|_2^2$. Define $\hat\va = \vm_{i^*}$. Now we have $\forall 1\leq i\leq n$,
\begin{align}
& \left\|\vm_i-\frac{w_i}{w_{i^*}}\hat\va\right\|_2^2\nonumber \\
\leq & 2\|\vm_i - w_i\va\|_2^2 + 2\left\|\frac{w_i}{w_{i^*}}\hat\va - w_i\va\right\|_2^2\label{eq:21} \\
= & 2\|\vm_i-w_i\va\|_2^2+2w_i^2\left\|\frac{\hat\va}{w_{i^*}}-\va\right\|_2^2\nonumber \\
\leq & 2\|\vm_i-w_i\va\|_2^2+2w_i^2\left\|\frac{\vm_i}{w_i}-\va\right\|_2^2\label{eq:22} \\
= & 4\|\vm_i-w_i\va\|_2^2.\label{eq:a}
\end{align}
Here (\ref{eq:21}) is by Fact \ref{fact:doubled_triangle}, and (\ref{eq:22}) is by the optimality of $i^*$ and $\hat\va = \vm_{i^*}$.

Let $S\subseteq\{1,\ldots,m\}$ be the support of $\hat\va=\vm_{i^*}$. Decompose $\vm_i$ as $\vm_i=\vm_i^1+\vm_i^2$ where $\vm_i^1$ is supported on $S$ and $\vm_i^2$ is supported on $\bar S$. Let $T$ be the support of $\vm_i^1$. Let $p\in[0,1]$ denote $\frac{|T|}{|S|}$. Define $\hat w_i=1$ if $p\geq \frac 12$ and $\hat w_i=0$ otherwise. Now we have 
\begin{align}
& \left\|\vm_i-\frac{w_i}{w_{i^*}}\hat\va\right\|_2^2\nonumber \\
= &\left\|\vm_i^1-\frac{w_i}{w_{i^*}}\hat\va\right\|_2^2+\|\vm_i^2\|_2^2\nonumber \\
= & |S|\left(p\left(1-\frac{w_i}{w_{i^*}}\right)^2+(1-p)\left(\frac{w_i}{w_{i^*}}\right)^2\right)+\|\vm_i^2\|_2^2\nonumber \\
\geq & |S|p(1-p)+\|\vm_i^2\|_2^2\label{eq:23}\\
\geq & \frac{1}{2}|S|\min\{p,1-p\}+\|\vm_i^2\|_2^2\nonumber\\
= & \frac 12 \|\vm_i^1-\hat w_i\hat \va\|_2^2+\|\vm_i^2\|_2^2\nonumber\\
\geq & \frac{1}{2}(\|\vm_i^1-\hat w_i\hat\va\|_2^2+\|\vm_i^2\|_2^2)\nonumber\\
= &\frac 12 \|\vm_i-\hat w_i\hat\va\|_2^2.\label{eq:w}
\end{align}
Here (\ref{eq:23}) is by Cauchy-Schwarz: $$\left(p(1-t)^2 + (1-p)t^2\right)((1-p) + p) \geq \left(\sqrt{p(1-p)}\cdot (1-t) + \sqrt{p(1-p)}\cdot t\right)^2 = p(1-p).$$

Combining (\ref{eq:a}) and (\ref{eq:w}), we have $\|\vm_i-\hat w_i\hat \va\|_2^2\leq 8\|\vm_i-w_i\va\|_2^2$. Moreover, it's clear that $\hat\va$ and $\hat\vw = (\hat w_1,\ldots,\hat w_n)^\top$ can be computed in poly-time.
\end{proof}

\section{Proof of Fact~\ref{fact:unit}}
\label{sec:proof-unit}
\begin{proof}
The claim holds trivially when either $\vx$ or $\vy$ is the zero vector. Now we consider $\vx$ being a unit vector and $\vy$ being a non-zero vector. Let $\alpha\in[0,\pi/2]$ denote the angle between $\vx$ and $\vy$. Note that $\|\vy - \theta\vx\|_2$ is at least the distance from point $\vy$ to the line defined by $\{\theta\vx:\theta\in\bb R\}$, so we have $\|\vy - \theta\vx\|_2\geq \|\vy\|_2\sin\alpha$. On the other hand, $\|\bar \vy-\vx\|_2^2 = 2 - 2\langle \bar\vy,\vx\rangle = 2-2\cos\alpha$. Therefore, the lemma reduces to
\begin{equation*}
    \sin^2\alpha\geq 1-\cos\alpha,
\end{equation*}
which is obviously true because $1 = \sin^2\alpha + \cos^2\alpha$ and $\cos\alpha\geq \cos^2\alpha$.
\end{proof}
\section{Proof of Theorem~\ref{thm:single}}
\label{sec:proof-single}
\begin{proof}
It is clear that $(A,W)$ is a feasible solution, and the computation we need besides the weighted $k$-means subroutine can be done in time linear in the number of entries in $M$, which justifies the claimed running time. We now prove that it achieves an objective at most $2r$ times that achieved by the optimal solution $(A^\opt, W^\opt)$. We can assume WLOG that each column $\va_s^\opt$ of $A^\opt$ is either a unit vector or the zero vector because we can scale up a column of $A$ and scale down the corresponding row of $W$ by the same factor without changing the product $AW$. Since the rows of $W^\opt$ are non-negative and orthogonal, they have disjoint supports, so there exists $\phi':\{1,\ldots,n\}\rightarrow \{1,\ldots,k\}$ and $\theta_1',\ldots,\theta_n'\in\bb R_{\geq 0}$ such that for all $i=1,\ldots,n$, the $i$-th column of $W^\opt$ is $\theta_i'\ve_{\phi'(i)}$.

Now we have
\begin{align}
& \|M - A^\opt W^\opt\|_F^2\nonumber \\
= & \sum_{i=1}^n\|\vm_i - \theta_i'\va_{\phi'(i)}\|_2^2\nonumber \\
\geq & \frac 12\sum_{i=1}^n\|\vm_i\|_2^2\|\bar \vm_i - \va_{\phi'(i)}\|_2^2\label{eq:1}\\
\geq & \frac 1{2r}\sum_{i=1}^n\|\vm_i\|_2^2\|\bar \vm_i - \vc_{\phi(i)}\|_2^2\label{eq:2}\\
= & \frac 1{2r}\sum_{i=1}^n\|\vm_i - \|\vm_i\|_2\vc_{\phi(i)}\|_2^2\nonumber \\
\geq & \frac 1{2r}\sum_{i=1}^n\|\vm_i - \theta_i \vc_{\phi(i)}\|_2^2\label{eq:9}\\
= & \frac 1{2r}\|M - AW\|_F^2.\nonumber
\end{align}
Here, (\ref{eq:1}) is by Fact \ref{fact:unit}, (\ref{eq:2}) is because $(\vc,\phi)$ is an $r$-approximate solution to (\ref{eq:k-means}), and (\ref{eq:9}) is because $\theta_i \in\argmin_\theta\|\vm_i - \theta\vc_{\phi(i)}\|_2^2$.
\end{proof}
\section{Proof of Lemma~\ref{lm:prune}}
\label{sec:proof-prune}
\begin{proof}
First, we show the following inequality for $q_j'$ instead of $q_j$: 
\begin{equation}
\label{eq:q'}
\sum_{j=1}^kq'_j\|\vc_j-\va_{\sigma(j)}\|_2^2\leq 8\sum_{j=1}^kq'_j\|\vc_j-\vz_{\sigma'(j)}\|_2^2.
\end{equation}
We start by constructing an alternative feasible solution to (\ref{eq:final_optimization}): $\va_1',\ldots,\va_k'\in\{\vzero, \vz_1,\ldots,\vz_{k_1}\}$. For any $\vc_j$ with $q'_j>0$, we say $j$ is ``matched'' if $\angle(\vc_j,\vz_{\sigma'(j)})< \frac \pi 6$. Note that if $\sigma(j_1) = \sigma(j_2)$ and they are both ``matched'', then we must have $\vz_{\sigma'(j_1)} = \vz_{\sigma'(j_2)}$, because otherwise $\frac\pi 2 = \angle (\vz_{\sigma'(j_1)}, \vz_{\sigma'(j_2)})\leq \angle (\vc_{j_1},\vz_{\sigma'(j_1)}) + \angle (\vc_{j_2},\vz_{\sigma'(j_2)}) + \angle (\vc_{j_1},\vc_{j_2})<\frac\pi 6 + \frac\pi 6 + \frac\pi 6$, a contradiction. Also, if $\sigma(j_1) \neq \sigma(j_2)$ and they are both ``matched'', then we must have $\vz_{\sigma'(j_1)} \neq \vz_{\sigma'(j_2)}$, because otherwise $\frac\pi 3<\angle(\vc_{j_1},\vc_{j_2})\leq \angle(\vc_{j_1}, \vz_{\sigma'(j_1)}) + \angle (\vc_{j_2},\vz_{\sigma'(j_2)})< \frac \pi 6 + \frac\pi 6$, a contradiction again. Therefore, we can uniquely define $\va'_s$ to be $\vz_{\sigma'(j)}$ whenever there exists a ``matched'' $j$ in $\sigma^{-1}(s)$, and we know different $s$ must correspond to different $\va'_s$. When such a ``matched'' $j$ in $\sigma^{-1}(s)$ does not exist, we simply define $\va'_s = \vzero$. Now $\va_1',\ldots,\va_k'$ are orthogonal to each other, so by the optimality of $\va_1,\ldots,\va_k$ in solving (\ref{eq:final_optimization}), we have 
\begin{equation*}
\sum_{j=1}^kq'_j\|\vc_j-\va_{\sigma(j)}\|_2^2\leq \sum_{j=1}^kq'_j\|\vc_j-\va'_{\sigma(j)}\|_2^2.
\end{equation*}
In order to prove (\ref{eq:q'}), we now only need to show that for every $j$ with $q_j'>0$, $\|\vc_j-\va'_{\sigma(j)}\|_2^2\leq 8 \|\vc_j-\vz_{\sigma'(j)}\|_2^2$. This is obviously true when $j$ is ``matched'' since $\va'_{\sigma(j)} = \vz_{\sigma'(j)}$. When $j$ is not ``matched'', we have $\angle (\vc_j, \vz_{\sigma'(j)})\geq \frac\pi 6$, so $\|\vc_j - \vz_{\sigma'(j)}\|_2^2\geq \sin^2\frac\pi 6= 1/4$, while $\|\vc_j-\va'_{\sigma(j)}\|_2^2 \leq \|\vc_j\|_2^2 + \|\va'_{\sigma(j)}\|_2^2 \leq 2$ by Fact \ref{fact:non-negative_triangle}. Therefore, $\|\vc_j-\va'_{\sigma(j)}\|_2^2\leq 8 \|\vc_j-\vz_{\sigma'(j)}\|_2^2$ is also true when $j$ is not ``matched''.

Now we prove 
\begin{equation}
\label{eq:q-q'}
\sum_{j=1}^k(q_j-q'_j)\|\vc_j-\va_{\sigma(j)}\|_2^2\leq \frac{2}{\sin^2(\pi/12)}\sum_{j=1}^k(q_j-q'_j)\|\vc_j-\vz_{\sigma'(j)}\|_2^2.
\end{equation}
We can decompose $q_j - q'_j$ by the iterations of the weight reduction step. Let $\bb I_t(j)$ denote the indicator for $\vc_j$ being chosen in the $t$-th iteration of weight reduction, and let $\Delta_t\geq 0$ denote the decrease in weight in the $t$-th iteration. We have $q_j - q'_j = \sum_{t} \bb I_t(j)\Delta_t$. Swapping sums, \eqref{eq:q-q'} is equivalent to
\begin{equation*}
\sum_t\Delta_t\sum_{j=1}^k\bb I_t(j)\|\vc_j-\va_{\sigma(j)}\|_2^2
\leq
\frac{2}{\sin^2(\pi/12)}
\sum_t\Delta_t\sum_{j=1}^k\bb I_t(j)\|\vc_j-\vz_{\sigma'(j)}\|_2^2.
\end{equation*}
Note that for a fixed $t$, $\bb I_t(j) = 1$ if and only if $j\in\{j_1,j_2\}$, where pair $(j_1,j_2)$ is selected in the $t$-th iteration of the weight reduction step. Thus, to prove \eqref{eq:q-q'}, it suffices
to prove that whenever $(j_1,j_2)$ is selected in the weight reduction step, we have $\|\vc_{j_1}-\va_{\sigma(j_1)}\|_2^2 + \|\vc_{j_2}-\va_{\sigma(j_2)}\|_2^2 \leq \frac{2}{\sin^2(\pi/12)}\left(\|\vc_{j_1}-\vz_{\sigma'(j_1)}\|_2^2 + \|\vc_{j_2}-\vz_{\sigma'(j_2)}\|_2^2\right)$.
Define $\alpha_1 := \angle (\vc_{j_1}, \vz_{\sigma'(j_1)})$ and $\alpha_2:= \angle (\vc_{j_2}, \vz_{\sigma'(j_2)})$. Since $\angle(\vc_{j_1},\vc_{j_2})\in[\pi/6,\pi/3]$, we always have $\alpha_1 + \alpha_2\geq \pi/6$, whether or not $\sigma'(j_1) = \sigma'(j_2)$. Therefore, we have $\|\vc_{j_1}-\vz_{\sigma'(j_1)}\|_2^2 + \|\vc_{j_2}-\vz_{\sigma'(j_2)}\|_2^2\geq \sin^2\alpha_1 + \sin^2\alpha_2\geq 2\sin^2\frac\pi {12}$ by the convexity and monotonicity of $\sin^2x$ over $[0,\pi/2]$. On the other hand, $\|\vc_{j_1}-\va_{\sigma(j_1)}\|_2^2 + \|\vc_{j_2}-\va_{\sigma(j_2)}\|_2^2\leq \|\vc_{j_1}\|_2^2 + \|\va_{\sigma(j_1)}\|_2^2 + \|\vc_{j_2}\|_2^2 + \|\va_{\sigma(j_2)}\|_2^2\leq 4$ by Fact \ref{fact:non-negative_triangle}. This concludes the proof of $\|\vc_{j_1}-\va_{\sigma(j_1)}\|_2^2 + \|\vc_{j_2}-\va_{\sigma(j_2)}\|_2^2 \leq \frac{2}{\sin^2(\pi/12)}\left(\|\vc_{j_1}-\vz_{\sigma'(j_1)}\|_2^2 + \|\vc_{j_2}-\vz_{\sigma'(j_2)}\|_2^2\right)$.

Combining (\ref{eq:q'}) with (\ref{eq:q-q'}) proves the Lemma.
\end{proof}
\section{Proof of Theorem~\ref{thm:k>=n}}
\label{sec:proof-k>=n}
\begin{proof}
By symmetry ($\|M-AW\|_F^2 = \|M^\top - W^\top A^\top\|_F^2$), we can assume WLOG that $k\geq n$. In this special case, the first step of the algorithm, solving weighted $k$-means, becomes trivial. We can simply choose $\phi(i) = i,\vc_i = \bar\vm_i$ and $q_i = \ell_i$. The second and the third steps remain the same.

To analyze the approximation ratio, we first write $\|M-AW\|_F^2$ as 
\begin{equation}
\label{eq:7}
\|M-AW\|_F^2 = \sum_{i=1}^n\|\vm_i - \theta_i \va_{\sigma(i)}\|_2^2 = \sum_{i=1}^n q_i\|\bar\vm_i - \bar\theta_i \va_{\sigma(i)}\|_2^2,
\end{equation}
where $\bar\theta_i = \left\{\begin{array}{ll}\frac{\langle \bar\vm_i,\va_{\sigma(i)}\rangle}{\|\va_{\sigma(i)}\|_2^2}, & \textup{if }\va_{\sigma(i)}\neq \vzero\\ 0,& \textup{if }\va_{\sigma(i)} = \vzero\end{array}\right.$.

Similarly to Lemma \ref{lm:prune}, we have the following lemma (proved in Appendix \ref{appendix}):
\begin{lemma}
\label{lm:prune'}
Let $\vz_1,\ldots,\vz_{k_1}\in\bb R^m_{\geq 0}$ be non-negative unit vectors that are orthogonal to each other. For any $\sigma':\{1,\ldots,n\}\rightarrow\{1,\ldots,k'\}$, we have
\begin{equation}
\label{eq:k>=n}
    \sum_{i=1}^n q_i\|\bar\vm_i - \bar \theta_i \va_{\sigma(i)}\|_2^2 \leq \frac{1}{\sin^2(\pi/12)}\sum_{i=1}^n q_i\|\bar \vm_i - \langle \bar \vm_i, \vz_{\sigma'(i)}\rangle \vz_{\sigma'(i)}\|_2^2.
\end{equation}
\end{lemma}

Suppose the optimal solution is $(A^\opt,W^\opt)$. Again, we remove the columns of $A^\opt$ filled with the zero vector and also remove the corresponding rows in $W^\opt$. Suppose the sizes of $A^\opt$ and $W^\opt$ now change to $m\times k_1$ and $k_1\times n$. We assume WLOG that every column $\va^\opt_s$ of $A^\opt$ is a unit vector. Suppose the $i$-th column of $W^\opt$ is $\theta_i'\ve_{\phi'(i)}$. We have
\begin{align}
    \|M-A^\opt W^\opt\|_F^2 = &\sum_{i=1}^n\|\vm_i - \theta_i'\va^\opt_{\phi'(i)}\|_2^2\nonumber\\ 
    \geq & \sum_{i=1}^n\|\vm_i - \langle \vm_i, \va^\opt_{\phi'(i)} \rangle\va^\opt_{\phi'(i)}\|_2^2\nonumber \\ 
    = & \sum_{i=1}^nq_i\|\bar\vm_i - \langle \bar\vm_i, \va^\opt_{\phi'(i)}\rangle\va^\opt_{\phi'(i)}\|_2^2.\label{eq:8}
\end{align}
Setting $\vz_{\sigma'(i)}$ in (\ref{eq:k>=n}) to be $\va^\opt_{\phi'(i)}$ and combining it with (\ref{eq:7}) and (\ref{eq:8}), we have $\|M-AW\|_F^2\leq \frac{1}{\sin^2(\pi/12)}\|M-A^\opt W^\opt\|_F^2$.
\end{proof}
\section{Proof of Lemma~\ref{lm:prune'}}
\label{appendix}
\begin{proof}
The proof is very similar to the proof of Lemma \ref{lm:prune}. First, we show the following inequality for $q_i'$ instead of $q_i$: 
\begin{equation}
\label{eq:q''}
\sum_{i=1}^nq'_i\|\bar \vm_i-\bar\theta_i \va_{\sigma(i)}\|_2^2\leq 8\sum_{i=1}^nq'_i\|\bar\vm_i-\langle \bar\vm_i, \vz_{\sigma'(i)}\rangle \vz_{\sigma'(i)}\|_2^2.
\end{equation}
Note that $\bar\theta_i\in\argmin_{\theta\geq 0}\|\bar\vm_i - \theta \va_{\sigma(i)}\|_2^2$. Therefore,
\begin{equation*}
    \sum_{i=1}^nq'_i\|\bar \vm_i-\bar\theta_i \va_{\sigma(i)}\|_2^2\leq \sum_{i=1}^nq'_i\|\bar\vm_i - \va_{\sigma(i)}\|_2^2.
\end{equation*}
We now construct an alternative feasible solution to (\ref{eq:final_optimization}): $\va_1',\ldots,\va_k'\in\{\vzero, \vz_1,\ldots,\vz_{k_1}\}$. For any $\bar\vm_i$ with $q'_i>0$, we say $i$ is ``matched'' if $\angle(\bar\vm_i,\vz_{\sigma'(i)})< \frac \pi 6$. Note that if $\sigma(i_1) = \sigma(i_2)$ and they are both ``matched'', then we must have $\vz_{\sigma'(i_1)} = \vz_{\sigma'(i_2)}$, because otherwise $\frac\pi 2 = \angle (\vz_{\sigma'(i_1)}, \vz_{\sigma'(i_2)})\leq \angle (\bar\vm_{i_1},\vz_{\sigma'(i_1)}) + \angle (\bar\vm_{i_2},\vz_{\sigma'(i_2)}) + \angle (\bar\vm_{i_1},\bar\vm_{i_2})<\frac\pi 6 + \frac\pi 6 + \frac\pi 6$, a contradiction. Also, if $\sigma(i_1) \neq \sigma(i_2)$ and they are both ``matched'', then we must have $\vz_{\sigma'(i_1)} \neq \vz_{\sigma'(i_2)}$, because otherwise $\frac\pi 3<\angle(\bar\vm_{i_1},\bar\vm_{i_2})\leq \angle(\bar\vm_{i_1}, \vz_{\sigma'(i_1)}) + \angle (\bar\vm_{i_2},\vz_{\sigma'(i_2)})< \frac \pi 6 + \frac\pi 6$, a contradiction again. Therefore, we can uniquely define $\va'_s$ to be $\vz_{\sigma'(i)}$ whenever there exists a ``matched'' $i$ in $\sigma^{-1}(s)$, and we know different $s$ must correspond to different $\va'_s$. When such a ``matched'' $i$ in $\sigma^{-1}(s)$ does not exist, we simply define $\va'_s = \vzero$. Now $\va_1',\ldots,\va_k'$ are orthogonal to each other, so by the optimality of $\va_1,\ldots,\va_k$ in solving (\ref{eq:final_optimization}), we have 
\begin{equation*}
\sum_{i=1}^nq'_i\|\bar\vm_i-\va_{\sigma(i)}\|_2^2\leq \sum_{i=1}^nq'_i\|\bar\vm_i-\va'_{\sigma(i)}\|_2^2\leq 2\sum_{i=1}^nq'_i\|\bar\vm_i-\langle \bar\vm_i, \va'_{\sigma(i)}\rangle\va'_{\sigma(i)}\|_2^2,
\end{equation*}
where the second inequality is by Fact \ref{fact:unit} and the fact that $\va'_{\sigma(i)}$ is either the zero vector or a unit vector.

In order to prove (\ref{eq:q''}), we now only need to show that for every $i$ with $q_i'>0$, 
\begin{equation}
\label{eq:5}
\|\bar\vm_i-\langle\bar\vm_i, \va'_{\sigma(i)} \rangle \va'_{\sigma(i)}\|_2^2\leq 4 \|\bar\vm_i-\langle\bar\vm_i, \vz_{\sigma'(i)} \rangle\vz_{\sigma'(i)}\|_2^2.
\end{equation}
This is obviously true when $i$ is ``matched'' since $\va'_{\sigma(i)} = \vz_{\sigma'(i)}$. When $i$ is not ``matched'', we have $\angle (\bar\vm_i, \vz_{\sigma'(i)})\geq \frac\pi 6$. Since $q_i\geq q_i'>0$, we know $\bar\vm_i$ is a unit vector rather than the zero vector, so $\|\bar\vm_i - \langle \bar\vm_i, \vz_{\sigma'(i)} \rangle\vz_{\sigma'(i)}\|_2^2\geq \sin^2\frac\pi 6= 1/4$, while $\|\bar\vm_i-\langle \bar\vm_i, \va'_{\sigma(i)} \rangle\va'_{\sigma(i)}\|_2^2 = \|\bar\vm_i\|_2^2 - \langle\bar\vm_i,\va'_{\sigma(i)} \rangle^2 \leq 1$. Therefore, (\ref{eq:5}) is also true when $j$ is not ``matched''.

Now we prove 
\begin{equation}
\label{eq:q-q''}
\sum_{i=1}^n(q_i-q'_i)\|\bar\vm_i-\bar\theta_i \va_{\sigma(i)}\|_2^2\leq \frac{1}{\sin^2(\pi/12)}\sum_{i=1}^n(q_i-q'_i)\|\bar\vm_i-\langle\bar\vm_i, \vz_{\sigma'(i)} \rangle\vz_{\sigma'(i)}\|_2^2.
\end{equation}
Similarly to how we proved \eqref{eq:q-q'}, it suffices to prove that whenever $(i_1,i_2)$ is selected in the weight reduction step, we have
\begin{align}
    & \|\bar\vm_{i_1}-\bar\theta_{i_1} \va_{\sigma(i_1)}\|_2^2 + \|\bar\vm_{i_2}-\bar\theta_{i_2} \va_{\sigma(i_2)}\|_2^2\nonumber\\
    \leq & \frac{1}{\sin^2(\pi/12)}\left(\|\bar\vm_{i_1}-\langle \bar\vm_{i_1}, \vz_{\sigma'(i_1)}\rangle \vz_{\sigma'(i_1)}\|_2^2 + \|\bar\vm_{i_2}-\langle\bar\vm_{i_2}, \vz_{\sigma'(i_2)}\rangle\vz_{\sigma'(i_2)}\|_2^2\right).\label{eq:6}
\end{align}
Note that here $\bar\vm_{i_1}$ and $\bar\vm_{i_2}$ are both unit vectors because otherwise they would have zero weights ($q_{i_1} = q_{i_1}' = 0$ or $q_{i_2} = q_{i_2}' = 0$) and wouldn't be selected in the weight reduction step. 

Define $\alpha_1 := \angle (\bar\vm_{i_1}, \vz_{\sigma'(i_1)})$ and $\alpha_2:= \angle (\bar\vm_{i_2}, \vz_{\sigma'(i_2)})$. Since $\angle(\bar\vm_{i_1},\bar\vm_{i_2})\in[\pi/6,\pi/3]$, we always have $\alpha_1 + \alpha_2\geq \pi/6$, whether or not $\sigma'(i_1) = \sigma'(i_2)$. Therefore, we have 
\begin{align*}
    \left(\|\bar\vm_{i_1}-\langle \bar\vm_{i_1}, \vz_{\sigma'(i_1)}\rangle \vz_{\sigma'(i_1)}\|_2^2 + \|\bar\vm_{i_2}-\langle\bar\vm_{i_2}, \vz_{\sigma'(i_2)}\rangle\vz_{\sigma'(i_2)}\|_2^2\right) & \geq \sin^2\alpha_1 + \sin^2\alpha_2\\
    & \geq 2\sin^2\frac\pi {12}
\end{align*}
by the convexity and monotonicity of $\sin^2x$ over $[0,\pi/2]$. On the other hand, since $\bar\theta_i\in\argmin_{\theta\geq 0}\|\bar\vm_i - \theta \va_{\sigma(i)}\|_2^2$, we know $\|\bar\vm_i - \bar\theta_i\va_{\sigma(i)}\|_2^2\leq \|\bar\vm_i\|_2^2$, so
\begin{align*}
    & \|\bar\vm_{i_1}-\bar\theta_{i_1} \va_{\sigma(i_1)}\|_2^2 + \|\bar\vm_{i_2}-\bar\theta_{i_2} \va_{\sigma(i_2)}\|_2^2\\
    \leq  & \|\bar\vm_{i_1}\|_2^2 + \|\bar\vm_{i_2}\|_2^2\\
    = & 2.
\end{align*}
This concludes the proof of (\ref{eq:6}). 

Combining (\ref{eq:q''}) with (\ref{eq:q-q''}) proves the Lemma.
\end{proof}
\section{Additional Experiments}
\label{sec:additional-exp}
\subsection{Experiments on Real-world Data}
We run our single-factor orthogonality algorithm on real-world datasets from \citet{Dua:2019}. 
Following the setting in \citet{asteris2015orthogonal}, we choose $k = 6$ and use the \emph{relative squared Frobenius error (RSFE)}  to measure the performance of our algorithm. Suppose $M$ is the data matrix and $A, W$ are the output of the algorithm, RSFE is defined as $\|M - AW\|_F^2 / \|M\|_F^2$. Note that the orthogonality constraint is posed on the left factor $A$, so we need to first transpose the data matrix before running our algorithm in Section \ref{sec:single}.
Our algorithm achieves similar RSFE compared to the best previous algorithm recorded in Table 2 of \citet{asteris2015orthogonal} on each dataset, and achieves slightly smaller (better) RSFE on datasets {\scshape Arcence Train} and {\scshape Mfeat Pix}.
\begin{table}[h]
\begin{center}
\begin{tabular}{ |c|c|c|} 
 \hline
 Dataset & RSFE of Our Algorithm & Smallest RSFE recorded in \citet{asteris2015orthogonal} \\ 
 \hline
{\scshape Amzn Com.\ Rev} & $0.0467$ & 0.0462 \\
{\scshape Arcence Train} & $\mathbf{0.0760}$ & 0.0788 \\
{\scshape Mfeat Pix} & $\mathbf{0.2382}$ & 0.2447\\
{\scshape Pems Train} & 0.1279 & 0.1278\\
{\scshape BoW:KOS} & 0.7685 & 0.7609\\
{\scshape BoW:NIPS} & 0.7386 & 0.7252\\
 \hline
\end{tabular}
\caption{Experimental results on real-world data in the single-factor orthogonality setting.}
\label{table}
\end{center}
\end{table}

\subsection{Experiments in the Double-factor Orthogonality Setting}
\label{sec:double-experiments}
\begin{figure}[h]
\centering
\includegraphics[width = 0.48\textwidth]{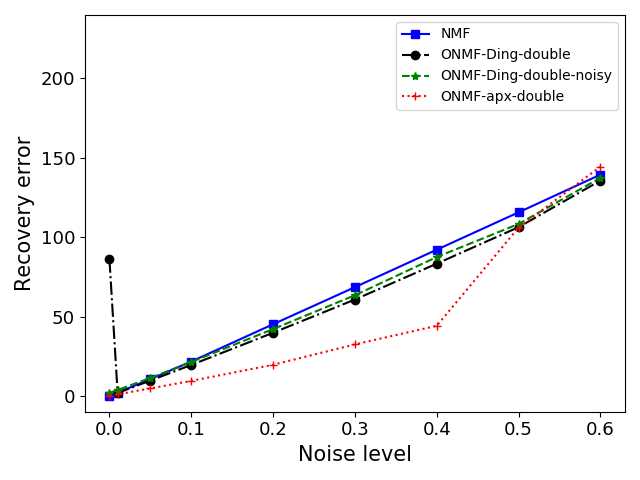}
\includegraphics[width = 0.48\textwidth]{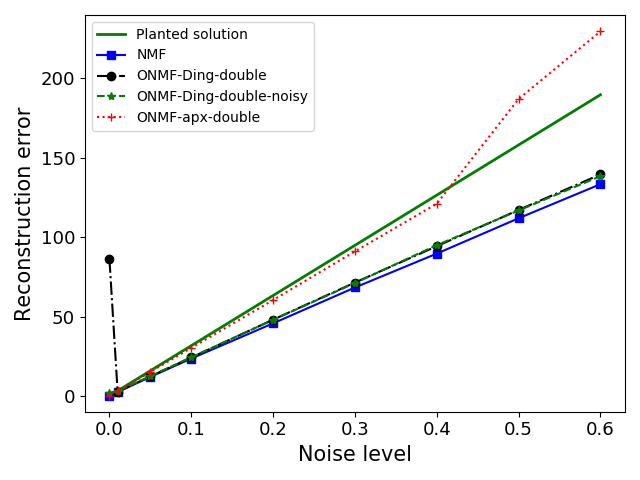}
\includegraphics[width = 0.48\textwidth]{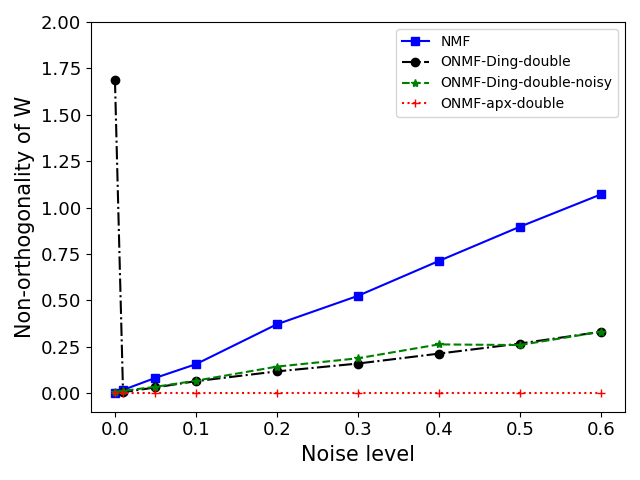}
\includegraphics[width = 0.48\textwidth]{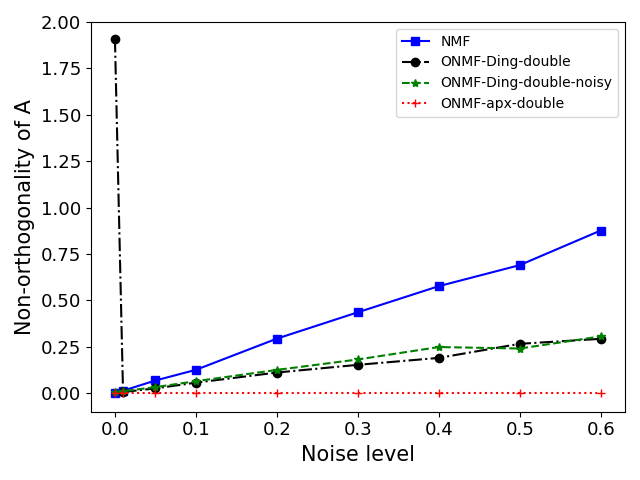}
\caption{Experimental results in the double-factor orthogonality setting.}
\label{fig:exp-double}
\end{figure}
We extend our experiments in Section \ref{sec:experiments} to the double-factor orthogonality setting, where we generate $A_\truth$ with orthogonal columns and $W_\truth$ with orthogonal rows. The only previous algorithm we know that handles the double-factor orthogonality is ONMF-Ding-double \citep{ding2006orthogonal}, which factorizes the input matrix $M\in \bb R^{m\times n}$ as the product of three non-negative matrices $M \approx ASW$ where $A\in\bb R^{m\times k}, S\in\bb R^{k\times k}, W\in\bb R^{k\times n}$, with the aim of making $A$ and $W$ satisfy the orthogonality constraint approximately. We compare our algorithm with ONMF-Ding-double together with the NMF algorithm \citep{lee2001algorithms} that does not aim for orthogonality. 
We keep other settings in Section \ref{sec:experiments} unchanged while choosing $m = 100, n = 500, k = 5$ so that ONMF-Ding-double can converge in a short time. While we run most algorithms 7 times and record the median results in Figure~\ref{fig:exp-double}, 
ONMF-Ding-double (resp.\ ONMF-Ding-double-noisy) is only run once at noise level 0.01 (resp.\ 0.00 and 0.01) because it took too long to finish.
As shown in Figure \ref{fig:exp-double}, our algorithm (ONMF-apx-double) is able to ensure perfect orthogonality for both factors and achieve better recovery error when the noise level is below 0.5.
We note that most of the reconstruction errors of the algorithms are below the reconstruction error of the planted solution $M_\truth$, which concentrates well around $10^{2.5} \approx 316$ times the noise level (thick green line in Figure~\ref{fig:exp-double}) by Fact~\ref{fact:planted-reconstruction}.
We also observe that ONMF-Ding-double takes more and more iterations to reach a solution with a reasonable approximation error as the noise level decreases towards zero, and it gets stuck at a suboptimal solution when the noise level is zero.
(Adding additional iid noise from the exponential distribution with mean $0.01$ to the input alleviates this issue, but that also slightly inflates the recovery error as shown by the green lines corresponding to ONMF-Ding-double-noisy in Figure \ref{fig:exp-double}.)

\subsection{Experiments for Different Inner Dimensions}
In experiment 1 (Section~\ref{sec:experiments}), we fixed $k = 10$ and studied how the performances of the algorithms vary with noise levels in the single-factor orthogonality setting. Now we fix the noise level to be $0.5$ and study the effect of different choices of the inner dimension $k$. We keep all other settings unchanged and record the results in Figure~\ref{fig:exp-k}. As in experiment 1, our algorithm (ONMF-apx) achieves perfect orthogonality with a significant improvement in the running time, and has smaller recovery errors than most previous algorithms. The experiment shows a common trend that the recovery (resp.\ reconstruction) error increases (resp.\ decreases) with the inner dimension $k$, although the amount of such change in the error is insignificant (note that the $y$-axes of the plots in the first row of Figure~\ref{fig:exp-k} do not start from zero). Of all algorithms studied in the experiment, the errors of our algorithm change the least with the inner dimension.
\begin{figure}[h]
\centering
\includegraphics[width = 0.48\textwidth]{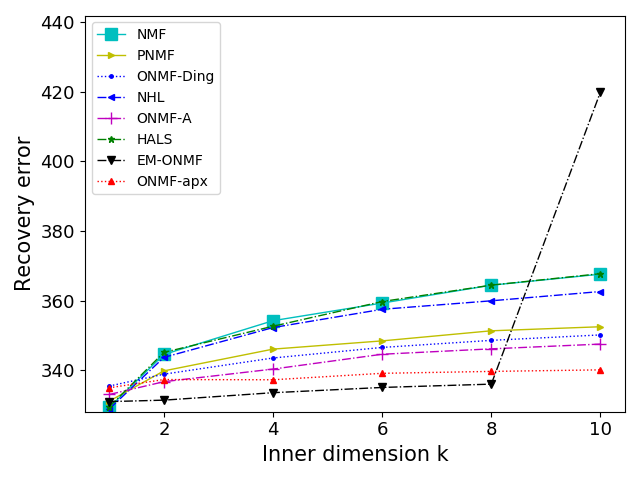}
\includegraphics[width = 0.48\textwidth]{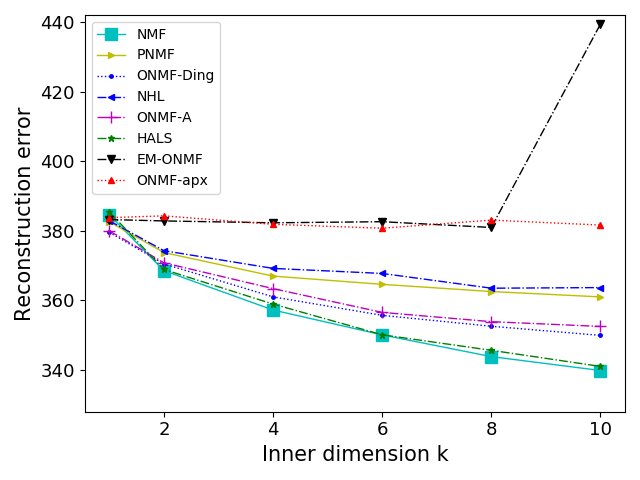}
\includegraphics[width = 0.48\textwidth]{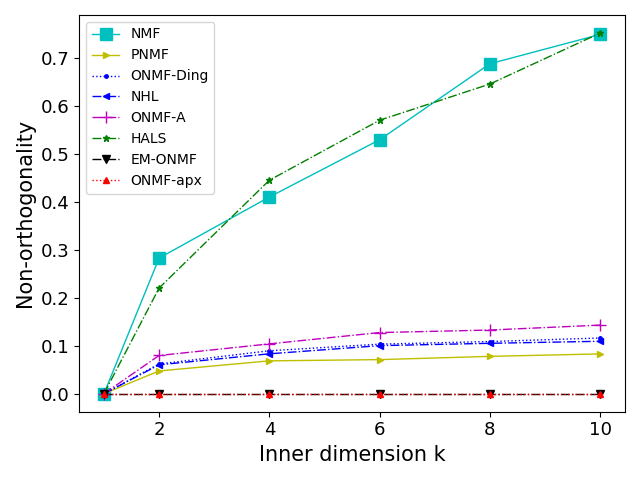}
\includegraphics[width = 0.48\textwidth]{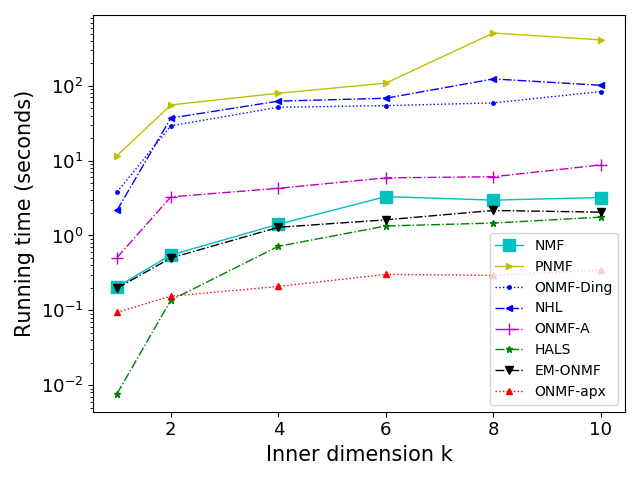}
\caption{Experimental results with different inner dimensions $k$ in the single-factor orthogonality setting.}
\label{fig:exp-k}
\end{figure}
\end{document}